\documentclass[a4paper,UKenglish,cleveref, autoref]{lipics-v2019}
\usepackage[utf8]{inputenc}
\usepackage{amsmath}
\usepackage{amssymb}
\usepackage{amsfonts}
\usepackage{amsthm}
\usepackage{mathtools}

\usepackage{tikz}
\usetikzlibrary{calc,positioning,decorations.pathreplacing}
\usepackage{environ}
\usepackage{comment}

\makeatletter
\newsavebox{\measure@tikzpicture}
\NewEnviron{scaletikzpicturetowidth}[1]{%
  \def\tikz@width{#1}%
  \begin{lrbox}{\measure@tikzpicture}%
  \BODY
  \end{lrbox}%
  \pgfmathparse{#1/\wd\measure@tikzpicture}%
  \BODY
}
\makeatother

\DeclareMathOperator*{\CH}{\mathrm{conv}} 
\DeclareMathOperator*{\supp}{supp} 

\newcommand{\oper}[1]{\textsc{#1}}
\newcommand{\atT}[2]{#1^{(#2)}}
\newcommand\numberthis{\addtocounter{equation}{1}\tag{\theequation}}

\nolinenumbers 

\hideLIPIcs  

\title{Computing Shapley Values for Mean Width in 3-D}
\author{Shuhao Tan}{Department of Computer Science, University of Maryland, College Park, USA \and \url{https://www.cs.umd.edu/~shuhao/}}{shuhao@cs.umd.edu}{}{}
\authorrunning{S.\,Tan}
\Copyright{Shuhao Tan}

\ccsdesc[500]{Theory of computation~Computational geometry}
\ccsdesc[300]{Theory of computation~Data structures design and analysis}

\keywords{Shapley value, mean width, dynamic convolution}
\acknowledgements{The author is very grateful to David Mount for discussions and advice for the paper. The author is also thankful to Geng Lin for proofreading and helping improving the readability.}
\begin{document}

\maketitle

\begin{abstract}
The Shapley value is a common tool in game theory to evaluate the importance of a player in a cooperative setting. In a geometric context, it provides a way to measure the contribution of a geometric object in a set towards some function on the set. Recently, Cabello and Chan (SoCG 2019) presented algorithms for computing Shapley values for a number of functions for point sets in the plane. More formally, a \emph{coalition game} consists of a set of players $N$ and a characteristic function $v: 2^N \to \mathbb{R}$ with $v(\emptyset) = 0$. Let $\pi$ be a uniformly random permutation of $N$, and $P_N(\pi, i)$ be the set of players in $N$ that appear before player $i$ in the permutation $\pi$. The Shapley value of the game is defined to be $\phi(i) = \mathbb{E}_\pi[v(P_N(\pi, i) \cup \{i\}) - v(P_N(\pi, i))]$. More intuitively, the Shapley value represents the impact of player $i$'s appearance over all insertion orders. We present an algorithm to compute Shapley values in 3-D, where we treat points as players and use the mean width of the convex hull as the characteristic function. Our algorithm runs in $O(n^3\log^2{n})$ time and $O(n)$ space. Our approach is based on a new data structure for a variant of the dynamic convolution problem $(u, v, p)$, where we want to answer $u\cdot v$ dynamically. Our data structure supports updating $u$ at position $p$, incrementing and decrementing $p$ and rotating $v$ by $1$. We present a data structure that supports $n$ operations in $O(n\log^2{n})$ time and $O(n)$ space. Moreover, the same approach can be used to compute the Shapley values for the mean volume of the convex hull projection onto a uniformly random $(d - 2)$-subspace in $O(n^d\log^2{n})$ time and $O(n)$ space for a point set in $d$-dimensional space ($d \geq 3$).
\end{abstract}

\section{Introduction}
Given a point set $P$ in $d$-dimensional space, many different functions can be applied to extract information about the set's geometric structure. Often, this involves properties of the convex hull of the set, such as its surface area and mean width. A natural question in this regard is the ``impact'' that any given point of $P$ has on the quantity of interest. One method for modeling the notion of impact arises from the context of cooperative games in game theory. In particular, Shapley values are a standard way to provide a distribution with \textit{fairness} in some sense. In this paper, we present an efficient algorithm to compute the Shapley values for points in 3-D where the payoff of a point set is defined as the mean width of its convex hull.

Formally, a \emph{coalition game} consists of a set of players $N$ and a characteristic function $v: 2^N \to \mathbb{R}$ with $v(\emptyset) = 0$. In our setting, we treat a point set $N\subset\mathbb{R}^3$ as the set of players and consider the characteristic function $v(Q)=w(\CH(Q))$ for $Q\subset N$, where $\CH(Q)$ denotes the convex hull of $Q$, and $w(\CH(Q))$ denotes the mean width of $\CH(Q)$. (Formal definitions will be given in \autoref{sec:prelim}.)

Let $\pi$ be a uniformly random permutation of $N$, and $P_N(\pi, i)$ be the set of players in $N$ that appear before player $i$ in the permutation $\pi$. The \emph{Shapley value} of player $i \in N$ is defined to be 
\begin{equation}
    \phi(i) = \mathbb{E}_\pi[v(P_N(\pi, i) \cup \{i\}) - v(P_N(\pi, i))].
\end{equation}
Intuitively, the Shapley value represents the expected marginal contribution of $i$ to the objective function over all permutations of $N$. There are wide applications of Shapley values. A survey by Winter \cite{WINTER20022025} and a book dedicated to this topic \cite{roth_1988} provide insights on how the concept can be interpreted and applied in multiple ways, such as utility of players, allocation of resources of the grand coalition and measure of power in a voting system. Moreover, the values can be characterized axiomatically, making it the only natural quantity that satisfies certain properties. More details can be found in standard game theory textbooks (\cite{Hart1989}, \cite[Chapter~5]{chakravarty2015course}, \cite[Section~9.4]{Myerson_game_theory_book}).

In convex geometry and measure theory, \emph{intrinsic volumes} are a key concept to characterize the ``size'' and ``shape'' of a convex body regardless of the translation and rotation in its underlying space. For example, Steiner's formula \cite{Gray2004} relates intrinsic volumes to the volume of the Minkowski-sum of a convex body and a ball. In general, one can define \emph{valuation} to be a class of  measure-like maps on open sets in a topological space. A formal definition of this concept can be found in \cite{Alvarez_valuation, Huber1993}. Functions such as volume and surface area fall into this class. Hadwiger's Theorem \cite{hadwiger2013vorlesungen, klain_1995_hadwiger} asserts that every continuous valuation  is a linear combination of intrinsic volumes. In $d$-dimensional space, the $d$-th, $(d-1)$-st and first intrinsic volumes are proportional to the usual Lebesgue measure, the surface area and the mean width, respectively \cite[Chapter~4]{schneider_2013}. Cabello and Chan \cite{Cabello19} presented efficient algorithms to compute Shapley values for area and perimeter for a point set in 2-D, which can be naturally extended to volume and surface area in 3-D. An algorithm that efficiently computes Shapley values for mean width in 3-D will then imply an algorithm that efficiently computes Shapley values for any continuous valuation in 3-D.

\paragraph*{Related work}
The problem of computing Shapley values for area functions on a point set was recently introduced by Cabello and Chan \cite{Cabello19}. They provided algorithms to compute Shapley values for area of convex hull, area of minimum enclosing disk, area of anchored rectangle, area of bounding box and area of anchored bounding box for a point set in 2-D. They also gave algorithms for perimeters by slightly modifying the algorithms. All the quantities they considered are defined in 2-D space. Although their algorithms naturally extend to higher dimension, there are interesting unique quantities in higher dimension that has yet been explored. The mean width considered in this paper is one such example.

Cabello and Chan also drew connections between computing Shapley values and stochastic computational geometry models. There have been many studies on the behavior of the convex hull under unipoint model where each point has an existential probability, for example see \cite{Agarwal2014, fink2016hyperplane, huang2016, Loffler2008, Suri2013, Jie2017_stochastic_CH}. In particular, Xue et al. \cite{Jie2017_stochastic_CH} discussed the expected diameter and width of the convex hull. Huang et al. \cite{huang2016} presented a way to construct $\epsilon$-coreset for directional width under the model. 

Mean width is often considered in the context of random polytopes in stochastic geometry. M{\"u}ller \cite{Muller1989} showed the asymptotic behavior of the mean width for a random polytope generated by sampling points on a convex body. B{\"o}roczyky \cite{BOROCZKY20092287} refined the result under the assumption that the convex body is smooth. Alonso-Guti{\'e}rrez and Prochno \cite{alonso2015gaussian} considered the case when the points are sampled inside an isotropic convex body. However, these results only consider the statistics when the number of points is large and the results are asymptotic. This paper views mean width on a computational perspective instead.

\paragraph*{Our contributions}
We show that the Shapley values for mean width of the convex hull for a point set in 3-D can be computed in $O(n^3\log^2{n})$ time and $O(n)$ space. We take a similar approach of computing area of convex hull in Cabello and Chan's paper, in the sense that we look at the incremental formation of convex hull and consider the contribution of individual geometric objects. In our case, we break down the mean width and express that in terms of quantities only related to edges and apply linearity of Shapley values.

A major difference is that the expression for mean width contains angle at the edge, making the calculation for the probability term for Shapley values depends on the intersection of two half-spaces as opposed to only one half-space. This prompts an expression that looks like convolution, but gets slightly changed when evaluated from one point to another. The setup can be captured as an instance of dynamic convolution, where one can change the convolution kernel at any position, and query any single position in the convoluted vector. Algebraic computation of this form was explored by Reif and Tate \cite{Reif97}. Frandsen et al. \cite{Frandsen01} gave a worst-case lower bound of $\Omega(\sqrt{n})$ time per operation for this problem. By exploiting the structure of sweeping by polar angle, we obtain a variant of dynamic convolution, where we have a pointer to the convolution kernel and another pointer to the convoluted result. We are only allowed to query at the pointer, update the convolution kernel at the pointer and move the pointers by one position. We present an online data structure for this variant that has $O(n\log^2{n})$ overall time for $n$ operations. There are some occurrences of algebraic computation in computational geometry, but many works \cite{Aronov2018, Cabello19, Deepak2007, Langerman2003} don't have such dynamic setting and rely mostly on computing a single convolution or multi-point evaluation of polynomials. Only a few (see, e.g., \cite{CHAN2010243}) employed a dynamic data structure. We believe our data structure is of independent interest for algorithms based on sweeping.

\section{Preliminaries}\label{sec:prelim}
\subparagraph*{Mean width}
The mean width of a compact convex body $X$ can be seen as the mean $1$-volume of $X$ projected on a uniformly random $1$-subspace. More formally, let $X\subset \mathbb{R}^d$ be a compact convex body. For $1\leq s\leq d$, the \emph{mean $s$-projection} of $X$ is defined as:
\begin{equation}
    M_s(X) = \int_{Q^s}{|X_u|f_s(u)\,du},
\end{equation}
where $Q^s$ is the set of $s$-subspace, $X_u$ is the projection of $X$ on $u$, $|\cdot|$ is the volume or the canonical measure in the underlying space, $f_s$ is the probability density function for uniformly sampling $s$-subspace from $Q^s$.

In this manner, the mean width of a point set $P$ can be defined as $M_1(\CH(P))$ where $\CH(P)$ is the convex hull of $P$.

It turns out that the mean $s$-projection of a convex polytope ($1\leq s\leq d-2$) can be decomposed into values only related to its $s$-facets.

Let $X\subset \mathbb{R}^d$ be a convex polytope. Let $V$ be a $s$-facet of $X$. Let $Q(V)$ be the set of half-spaces such that $V$ is on the hyper-planes defining the half-space. The \emph{exterior angle} of $V$ is defined as $\psi(V)=\frac{|\{q\in Q(V):X\subset q\}|}{|Q(V)|}$. Let $n(q)$ be the normal vector contained by the half-space $q$. The \emph{interior angle} of $V$ is defined as $\chi(V)=\frac{|\{q\in Q(V):\exists p\in V, \exists\epsilon > 0, p+\epsilon n(q)\in X|}{|Q(V)|}$.

These two definitions are generalizations of the angles formed by two adjacent edges of a convex polygon in 2-D. Let $p$ be a vertex of a convex polygon in 2-D with edges $e_1$, $e_2$ connected to it. Let $u$, $v$ be two unit vector parallel to $e_1$, $e_2$ respectively and away from $p$. The angle at $p$ is $\theta(p)=\arccos(u\cdot v)$. It is easy to see that the interior angle is indeed $\chi(p)=\theta(p)/2\pi$.

Moreover, when $s=d-2$, and $V\neq X$, we have $\chi(V)+\psi(V)=1/2$ for all $s$-facet $V$ of $X$.
\begin{lemma}[Miles \cite{Miles69}]\label{lemma:mean_content_formula}
Let $X\subset \mathbb{R}^d$ be a convex polytope, and $1\leq s\leq d-2$. Let $V_s(X)$ be the set of $s$-facet of $X$. We have:
\begin{equation}
    M_s(X)=\frac{\Gamma(\frac{s+1}{2})\Gamma(\frac{d-s+1}{2})}{\sqrt{\pi}\Gamma(\frac{d+1}{2})}\sum_{V\in V_s(X)}|V|\psi(V).
\end{equation}
\end{lemma}

\subparagraph*{Random permutations}
When considering random permutations, it is common to compute probabilities where constraints on the order of appearance of disjoint sets are imposed. The following lemma follows from simple counting where a proof can be found in \cite{Cabello19}.
\begin{lemma}\label{lemma:set_probability_formula}
Let $N$ be a set with $n$ elements. Let $\{x\}$, $A$ and $B$ be disjoint sets. The probability that all of $A$ appears before $x$ and all of $B$ appears after $x$ in a uniformly random permutation $\pi$ is $\frac{|A|!|B|!}{(|A|+|B|+1)!}$.
\end{lemma}

\subparagraph*{Assumptions}
We assume points in 3-D are of general positions where no three points are co-linear and no four points are co-planar. All points are assumed distinct. Throughout the paper, we often look at the projection of a 3-D point $p$ on a plane, we will still use the same symbol $p$ when looking on the plane for simplicity. We also assume a unit-cost real-RAM model of computation.

\section{Dynamic Convolution with Local Updates and Queries}
Before presenting the algorithm to compute mean width, we first present a data structure which is central to the algorithm.

We consider a variant of dynamic convolution where only local updates and queries are permitted. More formally, let $g:\mathbb{Z}\to\mathbb{R}$ be a fixed function where we can evaluate $g(x)$ in constant time for any $x$. Let $f:\mathbb{Z}\to\mathbb{R}$ initially be a zero function where we can change its values later. Let $p, c\in\mathbb{Z}$ be two variables initially $0$. We want to design a data structure that supports the following operations:

\begin{tabular}{llll}
    \oper{Query} & \oper{Update($x$)} & \oper{IncP} & \oper{DecP}  \\ \hline
     \quad Output $\sum_{i\in\mathbb{Z}} f(i + c)g(i)$ & \quad $f(p)\leftarrow f(p) + x$ & \quad $p\leftarrow p+1$& \quad $p\leftarrow p-1$\\[5pt]
    \oper{RotateLeft} & \oper{RotateRight}\\  \hline
    \quad $c\leftarrow c + 1$ & \quad $c\leftarrow c - 1$
\end{tabular}

Let $\atT{f}{i}$ be the $f$ after the $i$-th operation, so $\atT{f}{0}=0$. And similarly for $\atT{p}{i}$, and $\atT{c}{i}$. We can make the following observations:
\begin{lemma}\label{lemma:compact_support}
For any $0\leq i<j$, $\max\supp(\atT{f}{j}-\atT{f}{i}) - \min\supp(\atT{f}{j}-\atT{f}{i})\leq j-i$, where $\supp(f)$ is the support of $f$.
\end{lemma}
\begin{proof}
There are at most $j-i$ \oper{IncP} or \oper{DecP} operations between the $i$-th and the $j$-th operation, so $\max_{i\leq k\leq j}\atT{p}{k} - \min_{i\leq k\leq j}\atT{p}{k}\leq j-i$. And we have $\max\supp(\atT{f}{j}-\atT{f}{i})\leq \max_{i\leq k\leq j}\atT{p}{k}$ and $\min\supp(\atT{f}{j}-\atT{f}{i})\geq \min_{i\leq k\leq j}\atT{p}{k}$.
\end{proof}
\begin{corollary}\label{corollary:small_support}
For any $0\leq i<j$, $|\supp(\atT{f}{j}-\atT{f}{i})|\leq j-i$.
\end{corollary}
\begin{lemma}\label{lemma:limited_c}
For any $0\leq i<j$, $\max_{i\leq k\leq j}\atT{c}{k} - \min_{i\leq k\leq j}\atT{c}{k}\leq j-i$.
\end{lemma}
\begin{proof}
There are at most $j-i$ \oper{RotateLeft} or \oper{RotateRight} operations between the $i$-th and the $j$-th operation.
\end{proof}
\begin{lemma}\label{lemma:linearity_of_dot_product}
For any $0\leq i<j$, $\sum_k \atT{f}{j}(k+c)g(k)=\sum_k \atT{f}{i}(k+c)g(k) + \sum_k (\atT{f}{j}-\atT{f}{i})(k+c)g(k)$.
\end{lemma}
From \autoref{lemma:linearity_of_dot_product}, we can see that it is possible to break the operations into chunks, and only consider the changes on $f$ for each chunk. \autoref{lemma:compact_support} ensures that the actual differences are proportional to the size of a chunk. So the idea is: after a chunk of size $k$ is finished, we pre-compute all queries for the next $k$ operations by considering all possible changes in $c$. We incrementally build different levels of chunks so that after $k$ operations after a chunk of size $k$, we merge it with smaller chunks to build a larger chunk. \autoref{lemma:limited_c} and \autoref{corollary:small_support} ensure that the prediction only depends on a small set of values of $g$.

\begin{figure}[t]
    \centering
    \begin{tikzpicture}
    \tikzset{box/.style={rectangle,draw=black, very thick, minimum size=0.8cm}}
    
    \foreach \x in {1,...,8} {
        \node[box, fill=cyan!70] at (\x, 0) {$\x$};
    }
    \foreach \x in {9,...,12} {
        \node[box, fill=red!50] at (\x, 0) {$\x$};
    }
    \node[box, fill=yellow!50] at (13, 0) {$13$};
    
    \draw[decorate,decoration={brace,mirror}, thick] (.7, -.7) -- node[below]{$H(3)=\atT{f}{8}-\atT{f}{0}$} (8.3, -.7);
    \draw[decorate,decoration={brace,mirror}, thick] (8.7, -.7) -- node[below]{$H(2)=\atT{f}{12}-\atT{f}{8}$} (12.3, -.7);
    \draw[decorate,decoration={brace}, thick] (12.7, .7) -- node[above, xshift=-1cm]{$H(0)=\atT{f}{13}-\atT{f}{12}$} (13.3, .7);
    \end{tikzpicture}
    \caption{An example for $n=13=2^3+2^2+2^0$. Each box represents an operation.}
    \label{fig:dynamic_convolution_chunk}
\end{figure}
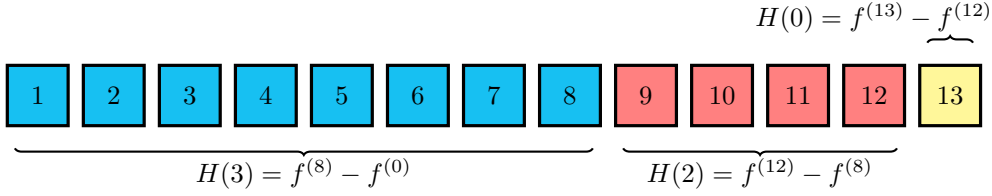

More specifically, we maintain chunks of changes based on the binary representation of the number of operations we have seen so far. Assume we've just performed the $k$-th operation. Let $b(k, i)$ be the $i$-th bit from right in the binary representation of $k$. We can write $k=\sum_{i:b(k,i)=1}2^i$ and we are going to use chunks of size $2^i$ for $b(k, i)=1$, where larger chunks are closer to the start of operations. Let $i_1 < i_2< \dots< i_q$ be the $i$'s such that $b(k, i)=1$, and $d(l)=\sum_{j=1}^{l}2^{i_j}$. Let $s(j)=k-d(j)$ and $e(j)=k-d(j-1)$. $(s(j), e(j)]$ represents the range of the $j$-th chunk. In other words, we split the operations into chunks of size $2^{i_q}, 2^{i_{q-1}},\dots, 2^{i_1}$. We maintain chunks of changes as $H(i_j)=\atT{f}{e(j)}-\atT{f}{s(j)}$. See \autoref{fig:dynamic_convolution_chunk} for an example. Moreover, for each chunk, we maintain an array of predicted queries $A(i_j)$ on the chunk $H(i_j)$, so that $A(i_j, 2^{i_j}+c)=\sum_p H(i_j)(p+\atT{c}{e(j)}+c)g(p)$ for $-2^{i_j}\leq c\leq 2^{i_j}$. We also maintain an array of $I$ so that $I(i_j)=\atT{c}{k}-\atT{c}{e(j)}$.

Assume we have $A$, $H$, and $I$ after $n$ operations, the $(n+1)$-st operation can be performed easily: \oper{Query} should output $\sum_p \atT{f}{n}(p + \atT{c}{n})g(p)$, and we have:
\begin{align*}
    \sum_p \atT{f}{n}(p + \atT{c}{n})g(p) &= \sum_p \left(\sum_{i_j:b(n, i_j)=1} \atT{f}{e(j)}-\atT{f}{s(j)}\right)(p + \atT{c}{n})g(p) \\
    &= \sum_{i_j:b(k, i_j)=1} \sum_p (\atT{f}{e(j)}-\atT{f}{s(j)})(p+\atT{c}{e(j)}+I(i_j))g(p)\\
    &= \sum_{i_j:b(k, i_j)=1} \sum_p H(i_j)(p+\atT{c}{e(j)}+I(i_j))g(p)\\
    &= \sum_{i_j:b(k, i_j)=1} A(i_j, 2^{i_j}+I(i_j)).
\end{align*}

So we can answer query by performing a summation over $A$. \oper{Update} is handled by simply documenting the change. \oper{IncP} and \oper{DecP} are handled by incrementing and decrementing $p$, respectively. \oper{RotateLeft} and \oper{RotateRight} are handled by incrementing and decrementing all entries of $I$, respectively.

After each operation, an additional maintaining step is performed to ensure $A$, $H$, $I$ contain the correct information for the following operations. Assume that we want to maintain the data structure after the $n$-th operation, we record the change in $f$ as $\Delta f=\atT{f}{n}-\atT{f}{n-1}$. Observe that $A(i)$ doesn't change if $b(n, i)=b(n-1, i)$. Let $i^*=\min\{i:b(n, i)=1\}$, then only $A(i)$'s such that $i\leq i^*$ change. Moreover, $A(i)$ becomes empty if $i<i^*$. Similar phenomenon can be seen in the increment of a binary counter. We expect:
\begin{align*}
    A(i^*, 2^{i^*}+c) &= \sum_p H(i^*)(p+\atT{c}{n}+c)g(p) \\
    &= \sum_p (\atT{f}{n}-\atT{f}{n-2^{i^*}})(p+\atT{c}{n}+c)g(p) \\
    &= \sum_p \left(\Delta f + \sum_{j=0}^{i^*-1}\atT{f}{n-2^j}-\atT{f}{n-2^{j+1}}\right)(p+\atT{c}{n}+c)g(p)\\
    &= \sum_p \left(\Delta f + \sum_{j=0}^{i^*-1}H(j)\right)(p+\atT{c}{n}+c)g(p).
\end{align*}

So we can merge $\Delta f$ and $\{H(j):j<i^*\}$ to get $H(j^*)=\atT{f}{n}-\atT{f}{n-2^{i^*}}$, and compute $A(i^*, 2^{i^*}+c)=\sum_p H(i^*)(p+\atT{c}{n}+c)g(p)$ for all $-2^{i^*}\leq c\leq 2^{i^*}$. This is almost a convolution. To see it clearer, let $L=\min\supp{H(i^*)}$ and $R=\max\supp{H(i^*)}$, we can rewrite the expression as $A(i^*, 2^{i^*}+c)=\sum_{p=L}^{R} H(i^*)(p)g(p-\atT{c}{n}-c)$. We now want to shift the origin so that $\supp H$ can start at $0$ to conform to the definition of a convolution. Let $H'(p)=H(i^*)(L+p)$ for $0\leq p\leq R-L$, and $H'(p)=0$ otherwise. Let $g'(p)=g(-p+L-\atT{c}{n}+2^{i^*+1})$, for $0\leq p \leq 3\cdot 2^{i^*}$, and $g'(p)=0$ otherwise. we have:
\begin{align*}
    A(i^*, 2^{i^*}+c) &= \sum_{p=L}^{R} H(i^*)(p)g(p-\atT{c}{n}-c)\\
    &= \sum_{p=0}^{R-L} H(i^*)(L+p)g(p-c+L-\atT{c}{n}) \\
    &= \sum_{p=0}^{R-L} H'(p)g'(2^{i^*+1}+c-p)\\
    &= H'*g'[2^{i^*+1}+c]
\end{align*}
Where $H'*g'$ is the discrete convolution of $H'$ and $g'$.

By \autoref{lemma:compact_support}, $R-L\leq 2^{i^*}$, so $0\leq 2^{i^*+1}+c-p\leq 3\cdot 2^{i^*}$. We can treat both $H'$ and $g'$ as circular vectors of size $3\cdot 2^{i^*}$, and compute $H'*g'[2^{i^*}+c]$ for all $-2^{i^*}\leq c\leq 2^{i^*}$ by the Fast Fourier Transform. We then use the result to construct $A(i^*)$. We set $I(i^*)=0$, and clear all $A(i)$, $H(i)$, and $I(i)$ where $i<i^*$.

Now that we have a way to maintain $A$, $H$, and $I$, the following Lemma shows that the \oper{Query} operation always succeed, i.e., $-2^i\leq I(i)\leq 2^i$.
\begin{lemma}\label{lemma:valid_range_c}
After the $n$-th operation, for any $i_j$ such that $b(n, i_j)=1$, $|\atT{c}{n}-\atT{c}{e(j)}|\leq 2^{i_j}$.
\end{lemma}
\begin{proof}
$|\atT{c}{n}-\atT{c}{e(j)}|\leq \max_{e(j)\leq k\leq n}\atT{c}{k} - \min_{e(j)\leq k\leq n}\atT{c}{k}\leq n-e(j)=d(j-1)\leq \sum_{k=0}^{i_j-1}2^k = 2^{i_j}$.
\end{proof}

We've proved the correctness of the data structure. We now show the complexity of it.
\begin{theorem}\label{thm:dynamic_convolution}
There is a data structure that supports $n$ operations for Dynamic Convolution with Local Updates and Queries in $O(n\log^2{n})$ time and $O(n)$ space.
\end{theorem}
\begin{proof}
Without considering the additional maintaining step after each operation, it is easy to see a single \oper{IncP}, \oper{DecP}, or \oper{Update} takes $O(1)$ time. Likewise a single \oper{RotateLeft}, \oper{RotateRight}, or \oper{Query} takes $O(\log{n})$ time since the number of bits in the binary representation of $n$ is $O(\log{n})$.

In the maintaining step, whenever the $i$-th bit in the binary representation of $n$ changes from $0$ to $1$, we need to merge changes of total size $O(2^i)$, run FFT on two arrays of size $3\cdot 2^i$, and clear histories of total size $O(2^i)$. These in all takes $O(2^i\log(2^i))$ time where the bottleneck is FFT. Like the analysis in binary counter, $i$-th bit flips from $0$ to $1$ in every $2^i$ operations, so the total time spent on maintaining step is: $\sum_{i=0}^{\log{n}} \frac{n}{2^i}O(2^i\log(2^i)) = O\left(n\sum_{i=0}^{\log{n}}i\log(2)\right)=O(n\log^2{n})$.

Adding two parts, the total time for $n$ operation is $O(n\log^2{n})$.

At any point, the space used by the data structure is linear to the size of $A$, $H$, $I$, even during the maintaining step. So the space complexity is $O\left(\sum_{i=0}^{\log{n}}2^i\right)=O(n)$.
\end{proof}
\begin{corollary}\label{corollary:dynamic_queue}
Let $g:\mathbb{Z}\to\mathbb{R}$ be a fixed function where we can evaluate $g(x)$ in constant time. Let $Q$ be a queue with elements in $\mathbb{R}$. There is a data structure that supports $n$ operations in $O(n\log^2{n})$ time and $O(n)$ space, where each operation is either pushing to the tail of $Q$, popping from the head of $Q$ or querying $\sum_{i=1}^{|Q|}Q(i)g(|Q|-i)$.
\end{corollary}
\begin{proof}
We use two instances of data structure $I_1=(f_1, g_1, p_1, c_1)$ and $I_2=(f_2, g_2, p_2, c_2)$ from \autoref{thm:dynamic_convolution}. We make $g_1(x)=g_2(x)=g(-x)$. Whenever we want to perform a push $x$, we perform \oper{Update($x$)}, \oper{RotateLeft} and \oper{IncP} on $I_1$, and perform \oper{RotateLeft} on $I_2$. Whenever we want to perform a pop $x$, we  perform \oper{Update($-x$)} and \oper{IncP} on $I_2$. Whenever we want to query, we query both $I_1$ and $I_2$ and return the sum. Each operation to the queue expands to constant number of operations on $I_1$ and $I_2$, so the running time is still $O(n\log^2{n})$ and the space is $O(n)$.
\end{proof}

\section{Computing Shapley Value of Mean Width in 3-D}
\subsection{Classification of Cases}
Let $X$ be a convex polyhedron in 3-D, \autoref{lemma:mean_content_formula} becomes:
\begin{equation}
    M_1(X)=\frac{1}{2}\sum_{e\in E(X)}l(e)\psi(e),
\end{equation}
where $E(X)$ is the set of edges of $X$ and $l(e)$ is the length of edge $e$.

Let $N\subset \mathbb{R}^3$ be a point set. Let $n$ be the number of points in the point set. The mean width we are considering is $M_1(\CH(N))$. Given a permutation $\pi$, and a point $p\in N$, for convenience we define:
\begin{equation*}
    C(\pi, p) = \CH(P_N(\pi, p)\cup\{p\}) \quad C'(\pi, p) = \CH(P_N(\pi, p)).
\end{equation*}
We then define the contribution of $p$ under permutation $\pi$ to be:
\begin{equation}
    \Delta(\pi, p) = M_1(C(\pi, p)) - M_1(C'(\pi, p)).
\end{equation}

When there are at least three points in $N$, we have $\psi(e)=1/2-\chi(e)$ for all $e\in E(N)$. In the case where $N$ is an edge $e$, we have $\psi(e)=1$. It is also clear that a single point has width $0$. So for $p\in N$, when considering the Shapley value $\phi(p)$, we can classify the permutations into three cases:
\begin{itemize}
    \item \textsc{Case 1}: There is one point before $p$ in the permutation. i.e., after inserting $p$, the point set forms a line segment.
    \item \textsc{Case 2}: There are two points before $p$ in the permutation. i.e., after inserting $p$, the point set forms a triangle.
    \item \textsc{Case 3}: There are three or more points before $p$ in the permutation.
\end{itemize}

In other words, we can write $\phi(p)$ as:
\begin{equation}
    \phi(p) = \sum_{\pi}\frac{\Delta(\pi, p)}{n!}=\sum_{\pi: \text{Case 1}}\frac{\Delta(\pi, p)}{n!} + \sum_{\pi: \text{Case 2}}\frac{\Delta(\pi, p)}{n!} + \sum_{\pi: \text{Case 3}}\frac{\Delta(\pi, p)}{n!}.
\end{equation}

We use a brute-force approach for Cases 1 and 2. For Case 1, we can rewrite the summation as:
\begin{align*}
    \sum_{\pi: \text{Case 1}}\frac{\Delta(\pi, p)}{n!} &= \sum_{\substack{q\in P\\q\neq p}}\sum_{\pi: \pi = (q, p, \dots)}\frac{\Delta(\pi, p)}{n!}\\
    &= \sum_{\substack{q\in P\\q\neq p}}\sum_{\pi: \pi = (q, p, \dots)}\frac{\|p-q\|/2}{n!}\\
    &= \sum_{\substack{q\in P\\q\neq p}} \frac{\|p-q\|}{2}\Pr(\pi=(q,p,\dots))\\
    &= \sum_{\substack{q\in P\\q\neq p}} \frac{\|p-q\|}{2}\frac{1}{n(n-1)}. \numberthis
\end{align*}

So for any $p$, we can compute the summation in $O(n)$ time by enumerating $q$. And it takes $O(n^2)$ in total to compute for every $p$.

For Case 2, we can rewrite the summation as:
\begin{align*}
    \sum_{\pi: \text{Case 2}}\frac{\Delta(\pi, p)}{n!} &= \sum_{\substack{q, r\in P\\p,q,r\text{ are distinct}}}\sum_{\pi: \pi = (q, r, p, \dots)}\frac{\Delta(\pi, p)}{n!}\\
    &= \sum_{\substack{q, r\in P\\p,q,r\text{ are distinct}}}\sum_{\pi: \pi = (q, r, p, \dots)}\frac{\frac{\|p-q\|+\|r-q\|+\|p-r\|}{4}-\frac{\|r-q\|}{2}}{n!}\\
    &= \sum_{\substack{q, r\in P\\p,q,r\text{ are distinct}}} \frac{\|p-q\|+\|p-r\|-\|r-q\|}{4}\Pr(\pi=(q,r,p,\dots))\\
    &= \sum_{\substack{q, r\in P\\p,q,r\text{ are distinct}}} \frac{\|p-q\|+\|p-r\|-\|r-q\|}{4}\frac{1}{n(n-1)(n-2)}. \numberthis
\end{align*}

Like Case 1, we can compute the summation in $O(n^2)$ time by enumerating $q$ and $r$. And it takes $O(n^3)$ in total to compute for every $p$.

In Case 3, each edge we are considering have two faces attached to it. In other words, each edge can be characterized by the common edge shared by two triangles formed by four points. We can denote an edge by $(q, r, t_1, t_2)$, where $(q, r)$ is the edge and is the common edge of $\bigtriangleup qrt_1$ and $\bigtriangleup qrt_2$. Without loss of generality, we assume unordered tuples when writing $(q, r)$ and $(t_1, t_2)$ to avoid double-counting. The exterior angle is completely determined by the quadruple. In case the convex hull has only 3 points, we allow $t_1=t_2$. We can then apply \autoref{lemma:mean_content_formula} and express $M_1(\CH(N'))$ for a point set $N'$ as:
\begin{align*}
    M_1(\CH(N')) &= \frac{1}{2}\sum_{e\in E(\CH(N'))}l(e)\psi(e) \\
    &= \frac{1}{2}\sum_{\substack{e=(q, r, t_1, t_2)\\q,r, t_1, t_2\in N'\\\bigtriangleup qrt_1, \bigtriangleup qrt_2 \text{ are faces of }\CH(N')}} l(e)\psi(e) \\
    &= \frac{1}{2}\sum_{\substack{e=(q, r, t_1, t_2)\\q,r, t_1, t_2\in N'}}\|q-r\| \psi(e) I_{\CH(N')}(q, r, t_1, t_2)
\end{align*}
where
\begin{equation*}
    I_{X}(q, r, t_1, t_2) = \begin{cases}
                            1 & \quad \bigtriangleup qrt_1, \bigtriangleup qrt_2 \text{ are faces of convex polyhedron }X\\
                            0 & \quad \text{otherwise}
                        \end{cases}
\end{equation*}
is an indicator variable.

Notice that $I_{C(\pi, p)}(q, r, t_1, t_2)=1$ implies $\pi$ being case~3. We then can write the summation as:
\begin{align*}
    \sum_{\pi: \text{Case 3}}\frac{\Delta(\pi, p)}{n!} = & \mathbb{E}_\pi[\Delta(\pi, p)|\pi:\text{Case 3}]\Pr(\pi:\text{Case 3})\\
    = & \frac{1}{2}\mathbb{E}_\pi\left[\sum_{\substack{e=(q, r, t_1, t_2)\\q,r, t_1, t_2\in N}}l(e)\psi(e) I_{C(\pi, p)}(q, r, t_1, t_2) - \right. \\
                & \left.\sum_{\substack{e=(q, r, t_1, t_2)\\q,r, t_1, t_2\in N}}l(e)\psi(e) I_{C'(\pi, p)}(q, r, t_1, t_2)
        \middle| \pi:\text{Case 3}\right]\Pr(\pi:\text{Case 3})\\
    = & \frac{1}{2}\left(\sum_{\substack{e=(q, r, t_1, t_2)\\q,r, t_1, t_2\in N}}l(e)\psi(e) \Pr(I_{C(\pi, p)}(q, r, t_1, t_2)=1\wedge \pi:\text{Case 3}) - \right.\\
    & \left. \sum_{\substack{e=(q, r, t_1, t_2)\\q,r, t_1, t_2\in N}}l(e)\psi(e) \Pr(I_{C'(\pi, p)}(q, r, t_1, t_2)=1\wedge \pi:\text{Case 3}) \right) \\
    = & \frac{1}{2}\sum_{\substack{q, r\in N}}\sum_{\substack{t_1, t_2\in N\\e=(q,r,t_1, t_2)}}l(e)\psi(e)\mathbb{E}_\pi[I_{C(\pi, p)}(q, r, t_1, t_2) - I_{C'(\pi, p)}(q, r, t_1, t_2)] \numberthis\label{eq:case3}
\end{align*}

Now consider the impact when inserting $p$ to $P_N(\pi, p)$. The convex hull doesn't change if $p$ is inside $\CH(P_N(\pi, p))$. Otherwise if we treat $\CH(P_N(\pi, p))$ as an opaque object, all the faces visible to $p$ will be removed in $\CH(P_N(\pi, p)\cup\{p\})$ and a pyramid-like cone with apex at $p$ will be added to $\CH(P_N(\pi, p)\cup\{p\})$. In terms of edges, there are three types of edges: edges removed, edges added and edges with angle changed. See \autoref{fig:edge_changes} for an illustration.
\begin{figure}[t]
    \centering
    \begin{tikzpicture}
    \tikzset{vert/.style={circle, inner sep=1pt, minimum size=1pt, draw=black, fill=black}}
    \tikzset{edge/.style={very thick, solid}}
    \tikzset{front/.style={opacity=0.9}}
    \tikzset{back/.style={opacity=0.3}}
    
    \node[vert, label=above:{$p$}] (p) at (4.16, -1.26) {};
    \node[vert] (v1) at (0.72, -3.19) {};
    \node[vert] (v2) at (3.26, -3.50) {};
    \node[vert] (v3) at (6.36, -3.03) {};
    \node[vert] (v4) at (5.91, -2.56) {};
    \node[vert, back] (v5) at (3.93, -2.43) {};
    \node[vert] (v6) at (0.91, -2.64) {};
    \node[vert] (v7) at (2.53, -2.56) {};
    \node[vert] (v8) at (3.65, -1.81) {};
    
    \draw[edge, front, draw=red] (p) -- (v1)
                           (p) -- (v2)
                           (p) -- (v3)
                           (p) -- (v4)
                           (p) -- (v6);
    \draw[edge, back, draw=red] (p) -- (v5);
    
    \draw[edge, front, draw=orange] (v6) -- (v1) -- (v2) -- (v3) -- (v4);
    \draw[edge, back, draw=orange] (v4) -- (v5) -- (v6);
    
    \draw[edge, back, draw=blue] (v8) -- (v1)
                           (v7) -- (v3)
                           (v8) -- (v3)
                           (v8) -- (v4)
                           (v8) -- (v5)
                           (v8) -- (v6)
                           (v7) -- (v8)
                           (v7) -- (v1)
                           (v7) -- (v2);
    \end{tikzpicture}
    \caption{Change of edges when $p$ is inserted, showing only the part visible to $p$. Red edges: edges added to the convex hull. Blue edges: edges removed from the convex hull. Orange edges: edges with angle changed.}
    \label{fig:edge_changes}
\end{figure}
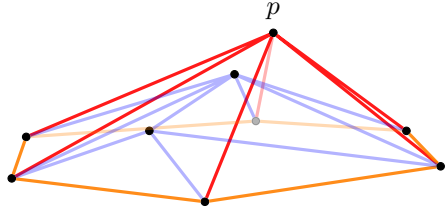

As we are using a $4$-tuple to represent an edge in the summation, edges with angles changed can be seen as removal and addition with different $(t_1, t_2)$. We have:
\begin{equation*}
   I_{C(\pi, p)}(q, r, t_1, t_2) - I_{C'(\pi, p)}(q, r, t_1, t_2) = \left\{ \begin{array}{lr}
                                                                        -1 & \quad (q, r, t_1, t_2)\text{ is an edge of $C'(\pi, p)$ and}\\ & \quad \text{visible to $p$}\\
                                                                        1 & \quad (q, r, t_1, t_2)\text{ is an edge of $C(\pi, p)$ and}\\ & \quad p\in \{t_1, t_2\}\\
                                                                        1 & \quad (q, r, t_1, t_2)\text{ is an edge of $C(\pi, p)$ and}\\ & \quad p\in \{q,r\}\\
                                                                        0 & \quad \text{otherwise}\\
                                                                    \end{array}\right.
\end{equation*}
The first three cases are correspondent to blue+orange, orange and red edges in \autoref{fig:edge_changes} respectively. So we can write:
\begin{equation}\label{eq:decompose_expectation}
\begin{aligned}
    &\mathbb{E}_\pi[I_{C(\pi, p)}(q, r, t_1, t_2) - I_{C'(\pi, p)}(q, r, t_1, t_2)]\\
    =& \Pr((q, r, t_1, t_2)\text{ is an edge of $C(\pi, p)$ and } p\in \{t_1, t_2\})\\
    & + \Pr((q, r, t_1, t_2)\text{ is an edge of $C(\pi, p)$ and } p\in \{q,r\})\\
    & - \Pr((q, r, t_1, t_2)\text{ is an edge of $C'(\pi, p)$ and visible to $p$})
\end{aligned}
\end{equation}
This gives us a way to split the final summation in \autoref{eq:case3} further into 3 summations. We will show how to compute them efficiently in the following subsection.

\subsection{Handling Case 3}
The idea is to enumerate edges $(q, r)$ and compute $\sum_{\substack{t_1, t_2\in N\\e=(q,r,t_1, t_2)}}l(e)\psi(e)\Pr(\cdot)$ for all $p$ where $\Pr(\cdot)$ is one of the three probabilities in \autoref{eq:decompose_expectation}.

\paragraph*{$(q, r, t_1, t_2)$ is an edge of $C'(\pi, p)$ and visible to $p$}
Given a pair of points $(q, r)$, we look at the projection of $N$ along the direction of $qr$.

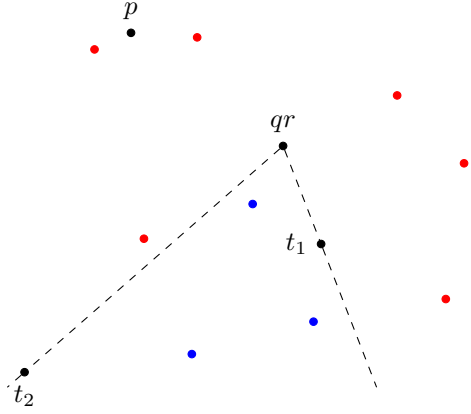
\begin{figure}[t]
    \centering
    \begin{tikzpicture}
    \tikzset{vert/.style={circle, inner sep=1pt, minimum size=1pt}}
    \tikzset{black/.style={draw=black, fill=black}}
    \tikzset{red/.style={draw=red, fill=red}}
    \tikzset{blue/.style={draw=blue, fill=blue}}
    
    \node[vert, black, label=above:{$p$}] (p) at (-2, 1.5) {};
    \node[vert, black, label=above:{$qr$}] (qr) at (0, 0) {};
    \node[vert, black, label=left:{$t_1$}] (t1) at (0.5, -1.3) {};
    \node[vert, black, label=below:{$t_2$}] (t2) at (-3.4, -3) {};
    \node[vert, red] (v1) at (2.38, -0.23) {};
    \node[vert, red] (v2) at (-1.13, 1.44) {};
    \node[vert, red] (v3) at (-1.83, -1.23) {};
    \node[vert, red] (v4) at (-2.48, 1.28) {};
    \node[vert, red] (v5) at (2.14, -2.03) {};
    \node[vert, red] (v6) at (1.5, 0.67) {};
    \node[vert, blue] (v7) at (-0.4, -0.77) {};
    \node[vert, blue] (v8) at (0.4, -2.33) {};
    \node[vert, blue] (v9) at (-1.2, -2.76) {};
    
    \draw[dashed] (qr) -- (t1) -- (1.23, -3.2)
                  (qr) -- (t2) -- (-3.627, -3.2);
    \end{tikzpicture}
    \caption{Projection of $N$ along $qr$. Edge $qr$ gets removed when forming $\CH(P_N(\pi, p)\cup\{p\})$ if and only if (a) $p$ is not in the cone formed by $qr$, $t_1$ and $t_2$, (b) $q$, $r$, $t_1$ and $t_2$ appear before $p$ in $\pi$ and (c) No point outside the cone (red point) appears before $p$ in $\pi$. }
    \label{fig:projection_removal}
\end{figure}
\begin{lemma}\label{lemma:removal_criteria}
$(q, r, t_1, t_2)$ is an edge of $C'(\pi, p)$ and visible to $p$ if and only if
\begin{enumerate}[(a)]
    \item $p$ is not in the cone formed by $qr$, $t_1$ and $t_2$. \label{item:visibility_criteria}
    \item $q$, $r$, $t_1$ and $t_2$ appear before $p$ in $\pi$. \label{item:exist_criteria}
    \item No point outside the cone appears before $p$ in $\pi$. \label{item:support_criteria}
\end{enumerate}
\end{lemma}
\begin{proof}
(\ref{item:exist_criteria}) is immediate as we need $q$, $r$, $t_1$ and $t_2$ to be in $P_N(\pi, p)$ so that the edge can be in $C'(\pi, p)$. Given (\ref{item:exist_criteria}), $(q, r, t_1, t_2)$ is an edge of $C'(\pi, p)$ if and only if $\bigtriangleup qrt_1$ and $\bigtriangleup qrt_2$ form two supporting planes of $P_N(\pi, p)$ if and only if no point outside the cone appears before $p$ in $\pi$. Finally, $(q, r, t_1, t_2)$ is visible to $p$ if and only if $pq$ and $pr$ are completely outside $C'(\pi, p)$ if and only if $p$ is not in the cone formed by $qr$, $t_1$ and $t_2$. See \autoref{fig:projection_removal} for an example.
\end{proof}
Treat $qr$ as the origin and let $p_1, p_2, \dots, p_{n-2}$ be the rest of the points in $N$ sorted by polar angles relative to $qr$. In other words, $p_1, p_2, \dots, p_{n-2}$ is the order of points when we sweep a ray starting from $qr$ around counterclockwise, initially to the direction of positive $x$-axis. Let $\theta_i$ be the polar angle of $p_i$. For convenience, we treat the sequence $p_1, p_2, \dots, p_{n-2}$ as a cyclic array, in the sense that $p_{n-1}=p_1$. Moreover, when we iterate through the sequence, $\theta_i$ is non-decreasing. In other words, when we iterate $p_1, p_2, \dots, p_{n-1}, p_n, \dots$, although $p_{n-1}$ and $p_1$ are the same point, we treat $\theta_{n-1}=\theta_1+2\pi$.

Let $W_{qr}(p_i, p_j)$ be the number of points in the cone formed by $qr$, $p_i$ and $p_j$. For a point $p$ outside the cone, \autoref{lemma:set_probability_formula} gives us: 
\begin{equation}\label{eq:probability_removal}
    \Pr((q, r, p_i, p_j)\text{ is an edge of $C'(\pi, p)$ and visible to $p$}) = \frac{4!(n-5-W_{qr}(p_i, p_j))!}{(n-W_{qr}(p_i, p_j))!}
\end{equation}
For simplicity, let
\begin{equation}\label{eq:g}
    g(i)=\frac{4!(n-5-i)!}{(n-i)!}
\end{equation}

Let $S(i)$ be the set of pairs $(p_j, p_k)$ such that $p_i$ is not in the cone formed by $qr$, $p_j$ and $p_k$. We assume $j\leq k$ and the cone is formed by sweeping from $p_j$ to $p_k$ counterclockwise. For a given point $p_i$, \autoref{lemma:removal_criteria} and \autoref{eq:probability_removal} gives:
\begin{align*}
    &\sum_{\substack{t_1, t_2\in N\\e=(q,r,t_1, t_2)}}l(e)\psi(e)\Pr((q, r, t_1, t_2)\text{ is an edge of $C'(\pi, p)$ and visible to $p_i$})\\
    =&\sum_{(p_j, p_k)\in S(i)}\|q-r\|\left(\frac{1}{2}-\frac{\theta_k-\theta_j}{2\pi}\right)g(W_{qr}(p_j, p_k))\\
    =& \|q-r\|\left(\sum_{(p_j, p_k)\in S(i)}\frac{1}{2}g(W_{qr}(p_j, p_k)) -\sum_{(p_j, p_k)\in S(i)} \frac{\theta_k-\theta_j}{2\pi}g(W_{qr}(p_j, p_k))\right). \numberthis\label{eq:split_to_si}
\end{align*}

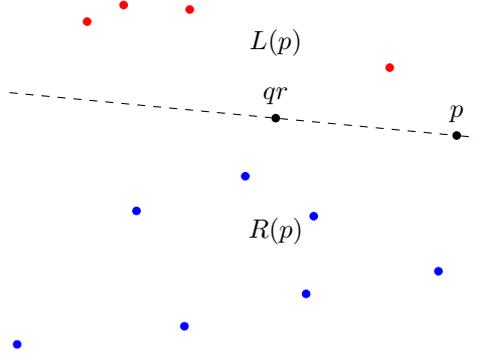
\begin{figure}[t]
    \centering
    \begin{tikzpicture}
    \tikzset{vert/.style={circle, inner sep=1pt, minimum size=1pt}}
    \tikzset{black/.style={draw=black, fill=black}}
    \tikzset{red/.style={draw=red, fill=red}}
    \tikzset{blue/.style={draw=blue, fill=blue}}
    \tikzset{purple/.style={draw=purple, fill=purple}}
    
    \node[vert, red] (v1) at (-2, 1.5) {};
    \node[vert, black, label=above:{$qr$}] (qr) at (0, 0) {};
    \node[vert, blue] (t1) at (0.5, -1.3) {};
    \node[vert, blue] (t2) at (-3.4, -3) {};
    \node[vert, black, label=above:{$p$}] (p) at (2.38, -0.23) {};
    \node[vert, red] (v2) at (-1.13, 1.44) {};
    \node[vert, blue] (v3) at (-1.83, -1.23) {};
    \node[vert, red] (v4) at (-2.48, 1.28) {};
    \node[vert, blue] (v5) at (2.14, -2.03) {};
    \node[vert, red] (v6) at (1.5, 0.67) {};
    \node[vert, blue] (v7) at (-0.4, -0.77) {};
    \node[vert, blue] (v8) at (0.4, -2.33) {};
    \node[vert, blue] (v9) at (-1.2, -2.76) {};
    
    \node at (0, 1) {$L(p)$};
    \node at (0, -1.5) {$R(p)$};
    
    \draw[dashed] (-3.5, 0.338) --  (qr) -- (p) -- (2.7, -0.261);
    \end{tikzpicture}
    \caption{Partition of points based on whether a point is left to $p-qr$ or right to $p-qr$. Red points are left, and blue points are right.}
    \label{fig:left_right_subset}
\end{figure}
We now show how to compute $\sum_{(p_j, p_k)\in S(i)} \frac{\theta_k-\theta_j}{2\pi}g(W_{qr}(p_j, p_k))$ for all $i$. For each point $p$, we partition the point set by whether a point is left to $p-qr$ or right to $p-qr$. More formally:%
\begin{equation}
    L(p)=\{p'\in N: (p-qr) \times (p'-qr)>0\} \quad R(p)=\{p'\in N: (p-qr) \times (p'-qr)<0\}
\end{equation}
as shown in \autoref{fig:left_right_subset}.

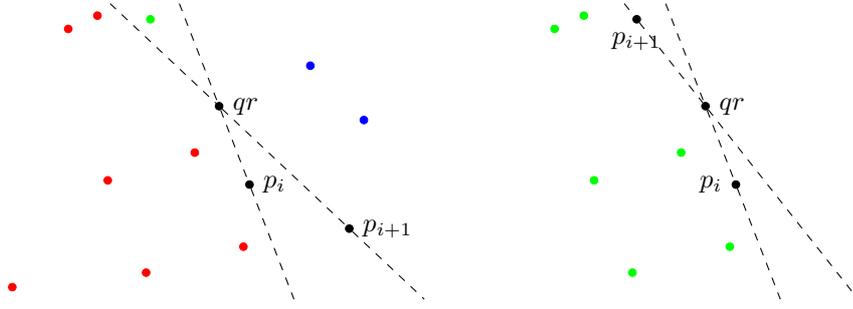
\begin{figure}[t]
    \begin{subfigure}[t]{0.45\textwidth}
    \centering
    \begin{tikzpicture}[scale=0.8]
    \tikzset{vert/.style={circle, inner sep=1pt, minimum size=1pt}}
    \tikzset{black/.style={draw=black, fill=black}}
    \tikzset{red/.style={draw=red, fill=red}}
    \tikzset{blue/.style={draw=blue, fill=blue}}
    \tikzset{green/.style={draw=green, fill=green}}
    
    \node[vert, red] (v1) at (-2, 1.5) {};
    \node[vert, black, label=right:{$qr$}] (qr) at (0, 0) {};
    \node[vert, black, label=right:{$p_i$}] (p) at (0.5, -1.3) {};
    \node[vert, red] (t2) at (-3.4, -3) {};
    \node[vert, blue] (v5) at (2.38, -0.23) {};
    \node[vert, green] (v2) at (-1.13, 1.44) {};
    \node[vert, red] (v3) at (-1.83, -1.23) {};
    \node[vert, red] (v4) at (-2.48, 1.28) {};
    \node[vert, black, label=right:{$p_{i+1}$}] (np) at (2.14, -2.03) {};
    \node[vert, blue] (v6) at (1.5, 0.67) {};
    \node[vert, red] (v7) at (-0.4, -0.77) {};
    \node[vert, red] (v8) at (0.4, -2.33) {};
    \node[vert, red] (v9) at (-1.2, -2.76) {};
    
    \draw[dashed] (-0.654, 1.7) -- (qr) -- (p) -- (1.231, -3.2)
                  (-1.792, 1.7) -- (qr) -- (np) -- (3.373, -3.2);
    \end{tikzpicture}
    
    \end{subfigure}
    \begin{subfigure}[t]{0.45\textwidth}
    \centering
    \begin{tikzpicture}[scale=0.8]
    \tikzset{vert/.style={circle, inner sep=1pt, minimum size=1pt}}
    \tikzset{black/.style={draw=black, fill=black}}
    \tikzset{red/.style={draw=red, fill=red}}
    \tikzset{blue/.style={draw=blue, fill=blue}}
    \tikzset{green/.style={draw=green, fill=green}}
    
    \node[vert, green] (v1) at (-2, 1.5) {};
    \node[vert, black, label=right:{$qr$}] (qr) at (0, 0) {};
    \node[vert, black, label=left:{$p_i$}] (p) at (0.5, -1.3) {};
    \node[vert, black, label=below:{$p_{i+1}$}] (np) at (-1.13, 1.44) {};
    \node[vert, green] (v3) at (-1.83, -1.23) {};
    \node[vert, green] (v4) at (-2.48, 1.28) {};
    \node[vert, green] (v7) at (-0.4, -0.77) {};
    \node[vert, green] (v8) at (0.4, -2.33) {};
    \node[vert, green] (v9) at (-1.2, -2.76) {};
    
    \draw[dashed] (-0.654, 1.7) -- (qr) -- (p) -- (1.231, -3.2)
                  (-1.334, 1.7) -- (np) -- (qr) -- (2.51, -3.2);
    \end{tikzpicture}
    \end{subfigure}
    \caption{When changing from $S(i)$ to $S(i+1)$, pairs formed between $p_{i+1}$ and $L(p_{i+1})\cup\{p_{i+1}\}$ are removed and pairs formed between $p_i$ and $R(p_i)\cup\{p_i\}$ are added. Left: $p_{i+1}\in L(p_i)$. Right: $p_{i+1}\in R(p_i)$.  }
    \label{fig:change_si}
\end{figure}

Consider the difference between $S(i)$ and $S(i)$. It is easy to see:
\begin{align*}
    \sum_{(p_j, p_k)\in S(i+1)} \frac{\theta_k-\theta_j}{2\pi}g(W_{qr}(p_j, p_k)) = & \sum_{(p_j, p_k)\in S(i)} \frac{\theta_k-\theta_j}{2\pi}g(W_{qr}(p_j, p_k))\\
    + &\sum_{p_j\in R(p_i)\cup\{p_i\}}\frac{\theta_i-\theta_j}{2\pi}g(W_{qr}(p_i, p_j)) \\
    - &\sum_{p_j\in L(p_{i+1})\cup\{p_{i+1}\}}\frac{\theta_j-\theta_{i+1}}{2\pi}g(W_{qr}(p_j, p_{i+1}))\numberthis\label{eq:delta_si}
\end{align*}
as demonstrated in \autoref{fig:change_si}.

We can further write:
\begin{align*}
    &\sum_{p_j\in R(p_i)\cup\{p_i\}}\frac{\theta_i-\theta_j}{2\pi}g(W_{qr}(p_i, p_j)) \\
    =& \theta_i\sum_{p_j\in R(p_i)\cup\{p_i\}}\frac{1}{2\pi}g(W_{qr}(p_i, p_j))-\sum_{p_j\in R(p_i)\cup\{p_i\}}\frac{\theta_j}{2\pi}g(W_{qr}(p_i, p_j)) \numberthis\label{eq:combination_eq}
\end{align*}
We will show how to compute:
\begin{equation*}
    \sum_{p_j\in R(p_i)\cup\{p_i\}}f(j)g(W_{qr}(p_i, p_j))
\end{equation*}
for an arbitrary function $f$ and for all $i$ at the same time. \autoref{eq:combination_eq} can then be computed by using $f(j)=1/2\pi$ and $f(j)=\theta_j/2\pi$, respectively.

If we sort the points in $R(p_i)\cup\{p_i\}$ by polar angles as $q_1, q_2, \dots, q_l$, where $l=|R(p_i)\cup\{p_i\}|$, it is easy to see $W_{qr}(p_i, q_j)=l-j$ for $1\leq j\leq l$. In other words, the number of points within the cone formed by $qr$, $p_i$, and $q_j$ is equal to the distance between $q_j$ and $p_i$ in the sequence $p_1$, $p_2$, \dots.

Now we consider an instance of the data structure from \autoref{corollary:dynamic_queue}. We use our $g$ from \autoref{eq:g} as the function used by the data structure. We start by choosing an arbitrary point $p_i$ and consider a ray opposite to $p_i-qr$. We sweep this ray counterclockwise until hitting $p_i$, for each point $p_j$ hit by the ray, we perform a push $f(j)$ to the data structure. After we hit $p_i$, we perform a query and the result is exactly $\sum_{p_j\in R(p_i)\cup\{p_i\}}f(j)g(W_{qr}(p_i, p_j))$. Next we sweep the ray to $p_{i+1}$ and pop all the points that are left to $p_{i+1}$. They should all appear at the head of the data structure. We then perform a query and the result is exactly $\sum_{p_j\in R(p_{i+1})\cup\{p_{i+1}\}}f(j)g(W_{qr}(p_i, p_j))$. We keep sweeping, popping and querying until coming back to $p_{i+n-3}$ which is $p_{i-1}$. During this process, each point gets pushed and popped at most twice and we perform $n-2$ queries. So the running time is $O(n\log^2{n})$ and the space complexity is $O(n)$ according to \autoref{corollary:dynamic_queue}.

Hence we can compute $\sum_{p_j\in R(p_i)\cup\{p_i\}}\frac{\theta_i-\theta_j}{2\pi}g(W_{qr}(p_i, p_j))$ for all $i$ in $O(n\log^2{n})$ time and $O(n)$ space. With the same idea but sweeping the other way around, we can also compute $\sum_{p_j\in L(p_{i+1})\cup\{p_{i+1}\}}\frac{\theta_j-\theta_{i+1}}{2\pi}g(W_{qr}(p_j, p_{i+1}))$ for all $i$ in the same time and space complexity.

Next we only need to compute $\sum_{(p_j, p_k)\in S(i)} \frac{\theta_k-\theta_j}{2\pi}g(W_{qr}(p_j, p_k))$ for a single $i$ and then used the pre-computed result to update to next $\sum_{(p_j, p_k)\in S(i+1)} \frac{\theta_k-\theta_j}{2\pi}g(W_{qr}(p_j, p_k))$ in constant time. Like in the previous case, we need to compute summations of the form $\sum_{(p_j, p_k)\in S(i)} f(j)g(W_{qr}(p_j, p_k))$. We again use an instance of data structure from \autoref{corollary:dynamic_queue} and use the same $g$ as above. We start by pushing $p_{i+1}$ and sweep counterclockwise. For each point hit, we first pop all the points that are left to the point, push the point to the data structure and finally perform a query. We stop after hitting $p_{i+n-3}$ which is $p_{i-1}$. The sum of all the queries will then be $\sum_{(p_j, p_k)\in S(i)} f(j)g(W_{qr}(p_j, p_k))$. In this procedure, each point is pushed and popped at most once, and $n-3$ queries are made. So it takes $O(n\log^2{n})$ time and $O(n)$ space to compute for a single $i$. After that it takes $O(n)$ time to compute for all $i$ by transitioning from $i$ to $i+1$ in constant time.

Using the same idea, we can compute the other part of \autoref{eq:split_to_si} in $O(n\log^2{n})$ time and $O(n)$ space as well.

For a fixed pair $(q, r)$, it takes $O(n\log{n})$ to sort other points by polar angles on the projected plane. And it takes $O(n)$ to pre-compute $g$. Hence, it takes $O(n\log^2{n})$ time and $O(n)$ space to compute $\sum_{\substack{t_1, t_2\in N\\e=(q,r,t_1, t_2)}}l(e)\psi(e)\Pr(\cdot)$ for all $p$ for the case where $(q, r, t_1, t_2)$ is edge of $C'(\pi, p)$ and visible to $p$. And it takes $O(n^3\log^2{n})$ time and $O(n)$ space overall by enumerating all pairs of $(q, r)$.

\paragraph*{$(q, r, t_1, t_2)$ is an edge of $C(\pi, p)$ and $p\in \{t_1, t_2\}$}
Again we fix a pair $(q, r)$, and look at the projection of $N$ along the direction of $qr$. Without loss of generality, assume $p=t_1$ in this case.
\begin{lemma}\label{lemma:addition_criterion_not_on_edge}
$(q, r, p, t_2)$ is edge of $C(\pi, p)$ if and only if
\begin{enumerate}[(a)]
    \item $q$, $r$ and $t_2$ appear before $p$ in $\pi$.
    \item No point outside the cone formed by $qr$, $p$ and $t_2$ appears before $p$ in $\pi$.
\end{enumerate}
\end{lemma}

The proof will be almost the same as the proof for \autoref{lemma:removal_criteria}. See \autoref{fig:projection_addition_not_on_edge} for an illustration.

\begin{figure}[t]
    \centering
    \begin{tikzpicture}
    \tikzset{vert/.style={circle, inner sep=1pt, minimum size=1pt}}
    \tikzset{black/.style={draw=black, fill=black}}
    \tikzset{red/.style={draw=red, fill=red}}
    \tikzset{blue/.style={draw=blue, fill=blue}}
    
    \node[vert, black, label=above:{$p$}] (p) at (-2, 1.5) {};
    \node[vert, black, label=above:{$qr$}] (qr) at (0, 0) {};
    \node[vert, red] (t1) at (0.5, -1.3) {};
    \node[vert, black, label=below:{$t_2$}] (t2) at (-3.4, -3) {};
    \node[vert, red] (v1) at (2.38, -0.23) {};
    \node[vert, red] (v2) at (-1.13, 1.44) {};
    \node[vert, blue] (v3) at (-1.83, -1.23) {};
    \node[vert, blue] (v4) at (-2.48, 1.28) {};
    \node[vert, red] (v5) at (2.14, -2.03) {};
    \node[vert, red] (v6) at (1.5, 0.67) {};
    \node[vert, red] (v7) at (-0.4, -0.77) {};
    \node[vert, red] (v8) at (0.4, -2.33) {};
    \node[vert, red] (v9) at (-1.2, -2.76) {};
    
    \draw[dashed] (qr) -- (p) -- (-2.267, 1.7)
                  (qr) -- (t2) -- (-3.627, -3.2);
    \end{tikzpicture}
    \caption{Projection of $N$ along $qr$. Edge $qr$ gets added when forming $\CH(P_N(\pi, p)\cup\{p\})$ if and only if (a) $q$, $r$ and $t_2$ appear before $p$ in $\pi$ (b) No point outside the cone formed by $qr$, $p$ and $t_2$ (red point) appears before $p$ in $\pi$. }
    \label{fig:projection_addition_not_on_edge}
\end{figure}
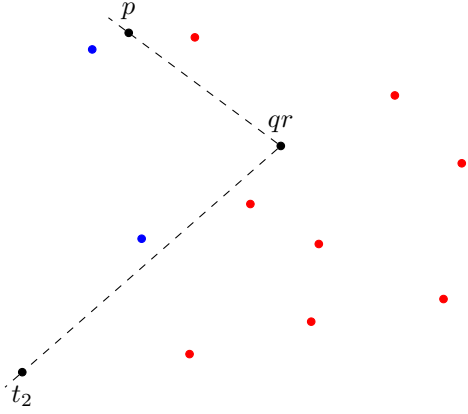

In this case
\begin{equation}
    \Pr((q, r, p, t_2)\text{ is an edge of $C(\pi, p)$}) = \frac{3!(n-4-W_{qr}(p, t_2))!}{(n-W_qr(p,t_2))!}
\end{equation}
We have a slightly different $g(i)=\frac{3!(n-4-i)!}{(n-i)!}$. Any point except $p$ can form a cone with $p$. So we can write
\begin{align*}
    &\sum_{\substack{t_2\in N\\e=(q,r,p_i, t_2)}}l(e)\psi(e)\Pr((q, r, p, t_2)\text{ is an edge of $C(\pi, p)$ })\\
    =&\sum_{p_j\in R(p_i)}\|q-r\|\left(\frac{1}{2}-\frac{\theta_i-\theta_j}{2\pi}\right)g(W_{qr}(p_i, p_j))\\
    +& \sum_{p_j\in L(p_i)}\|q-r\|\left(\frac{1}{2}-\frac{\theta_j-\theta_i}{2\pi}\right)g(W_{qr}(p_j, p_i))\numberthis\label{eq:addition_split}
\end{align*}

In the previous case, we've shown how to compute summations with very similar form as summations in \autoref{eq:addition_split}. The only differences are that we take $p_j\in R(p_i)$ instead of $p_j\in R(p_i)\cup\{p_i\}$ and we have a slightly different $g$ here. But clearly it take constant time to calculate the difference if we use the approach before here. Hence it takes $O(n^3\log^2{n})$ time in total and $O(n)$ space by using the same method.

\paragraph*{$(q, r, t_1, t_2)$ is an edge of $C(\pi, p)$ and $p\in \{q,r\}$}
Without loss of generality, assume $p=q$. We look at the projection of $N$ along the direction of $pr$.
\begin{lemma}\label{lemma:addition_criterion_on_edge}
$(p, r, t_1, t_2)$ is edge of $C(\pi, p)$ if and only if
\begin{enumerate}[(a)]
    \item $r$, $t_1$ and $t_2$ appear before $p$ in $\pi$.
    \item No point outside the cone formed by $pr$, $t_1$ and $t_2$ appears before $p$ in $\pi$.
\end{enumerate}
\end{lemma}

The proof will be almost the same as the proof for \autoref{lemma:removal_criteria}. See \autoref{fig:projection_addition_on_edge} for an illustration.

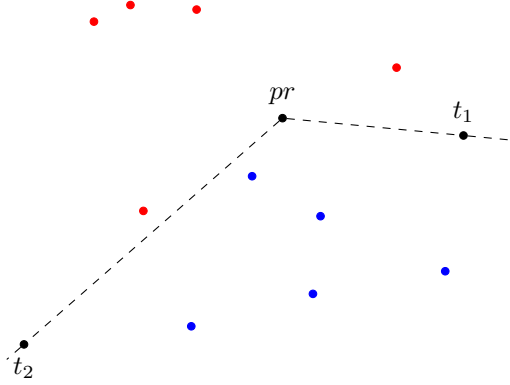
\begin{figure}[t]
    \centering
    \begin{tikzpicture}
    \tikzset{vert/.style={circle, inner sep=1pt, minimum size=1pt}}
    \tikzset{black/.style={draw=black, fill=black}}
    \tikzset{red/.style={draw=red, fill=red}}
    \tikzset{blue/.style={draw=blue, fill=blue}}
    
    \node[vert, red] (p) at (-2, 1.5) {};
    \node[vert, black, label=above:{$pr$}] (pr) at (0, 0) {};
    \node[vert, blue] (v1) at (0.5, -1.3) {};
    \node[vert, black, label=below:{$t_2$}] (t2) at (-3.4, -3) {};
    \node[vert, black, label=above:{$t_1$}] (t1) at (2.38, -0.23) {};
    \node[vert, red] (v2) at (-1.13, 1.44) {};
    \node[vert, red] (v3) at (-1.83, -1.23) {};
    \node[vert, red] (v4) at (-2.48, 1.28) {};
    \node[vert, blue] (v5) at (2.14, -2.03) {};
    \node[vert, red] (v6) at (1.5, 0.67) {};
    \node[vert, blue] (v7) at (-0.4, -0.77) {};
    \node[vert, blue] (v8) at (0.4, -2.33) {};
    \node[vert, blue] (v9) at (-1.2, -2.76) {};
    
    \draw[dashed] (pr) -- (t1) -- (3, -0.29)
                  (pr)-- (t2) -- (-3.627, -3.2);
    \end{tikzpicture}
    \caption{Projection of $N$ along $pr$. Edge $pr$ gets added when forming $\CH(P_N(\pi, p)\cup\{p\})$ if and only if (a) $r$, $t_1$ and $t_2$ appear before $p$ in $\pi$ (b) No point outside the cone formed by $pr$, $t_1$ and $t_2$ (red point) appears before $p$ in $\pi$. }
    \label{fig:projection_addition_on_edge}
\end{figure}

In this case
\begin{equation}
    \Pr((p, r, t_1, t_2)\text{ is an edge of $C(\pi, p)$}) = \frac{3!(n-4-W_{qr}(t_1, t_2))!}{(n-W_qr(t_1,t_2))!}
\end{equation}
And we use $g(i)=\frac{3!(n-4-i)!}{(n-i)!}$. In this case, any pair $(t_1, t_2)$ with $t_1\neq t_2$ will contribute to result. So we can write
\begin{align*}
    &\sum_{\substack{t_1, t_2\in N\\e=(p,r,t_1, t_2)}}l(e)\psi(e)\Pr((p, r, t_1, t_2)\text{ is an edge of $C(\pi, p)$})\\
    =&\frac{1}{2}\sum_{p_k}\sum_{p_j\in R(p_k)}\|p-r\|\left(\frac{1}{2}-\frac{\theta_k-\theta_j}{2\pi}\right)g(W_{qr}(p_k, p_j))\\
    +&\frac{1}{2}\sum_{p_k}\sum_{p_j\in L(p_k)}\|p-r\|\left(\frac{1}{2}-\frac{\theta_j-\theta_k}{2\pi}\right)g(W_{qr}(p_j, p_k))\numberthis
\end{align*}
We have a factor of $\frac{1}{2}$ because each pair is counted twice. In the previous case, all the inner summations have been pre-computed. So for a fixed $(p, r)$, the summation can be computed in $O(n)$ time given previous computation. Hence in total it takes $O(n^3)$ time and no additional space to compute this case.

Having resolved all the cases, we present our main theorem:
\begin{theorem}
Shapley values for mean width for a point set in 3-D can be computed in $O(n^3\log^2{n})$ time and $O(n)$ space.
\end{theorem}

\section{Discussion}
We have presented an algorithm to compute Shapley values in game theory, with respect to a point set in 3-D and the mean width of its convex hull. We provided an efficient algorithm based on a data structure for a variant of dynamic convolution. We believe the data structure may be of independent interest.

Our algorithm naturally extends to higher dimension to compute Shapley values for $M_{d-2}(\CH(P))$ for a $d$-dimensional point set $P$. This relies on the fact that the orthogonal space of a $(d-2)$-facet is a plane. In general, it takes $O(n^d\log^2{n})$ time to compute the Shapley values. It is also known that for a convex polytope $X$, $M_{d-1}(X)$ is equivalent to the $(d-1)$-volume of the boundary of $X$ up to some constant \cite{Miles69}. Hence the algorithm by Cabello and Chan \cite{Cabello19} with natural extension can be used to compute Shapley values for $M_{d-1}(\CH(P))$ in $O(n^d)$ time. It would be natural to ask whether there are efficient algorithms to compute $M_i(\CH(P))$ in general.

We can also consider $\epsilon$-coreset of Shapley values for geometric objects. Let $P$ be a set of geometric objects and $v$ be a characteristic function on $P$. We can define the $\epsilon$-coreset $\widetilde{P}$ to be a weighted set of geometric objects such that, for any geometric objects $x$, $(1-\epsilon)\phi_{P\cup\{x\}}(x)\leq \phi_{\widetilde{P}\cup\{x\}}(x) \leq (1+\epsilon)\phi_{P\cup\{x\}}(x)$ where $\phi_N$ means that the underlying player set for the Shapley value is $N$. Intuitively, $\phi_{P\cup\{x\}}(x)$ means the contribution $x$ makes when $x$ is added as an additional player. Does $\widetilde{P}$ exist? If so, what is the upper and lower bound of its size? How fast can we find a coreset?

\bibliography{ref}
\end{document}